\documentclass[reprint,amsfonts, amssymb, amsmath,  showkeys,pra, superscriptaddress, twocolumn,longbibliography,nofootinbib]{revtex4-2}

\usepackage{xcolor}
\usepackage{float}
\makeatletter
\let\newfloat\newfloat@ltx
\makeatother
\usepackage[english]{babel}
\usepackage[utf8]{inputenc}
\usepackage{graphics}
\usepackage{selinput}
\usepackage[normalem]{ulem}
\usepackage[shortlabels]{enumitem}

\usepackage{braket}
\usepackage{amsthm}
\usepackage{mathtools}
\usepackage{physics}
\usepackage{graphicx}
\usepackage[left=16mm,right=16mm,top=35mm,columnsep=15pt]{geometry} 
\usepackage{adjustbox}
\usepackage{placeins}
\usepackage[T1]{fontenc}
\usepackage{lipsum}
\usepackage{csquotes}
\usepackage{bm}
\usepackage{mathrsfs}

\usepackage[linesnumbered,ruled,vlined]{algorithm2e}
\SetKwInput{kwInit}{Init}

\def\HC{\mathcal{H}}

\def\ad{^{\dagger}}

\newcommand{\fsnull}[1]{}
\newcommand{\old}[1]{}

\usepackage[makeroom]{cancel}
\usepackage[toc,page]{appendix}
\usepackage[colorlinks=true,citecolor=blue,linkcolor=magenta]{hyperref}

\usepackage{tikz}
\tikzset{every picture/.style=remember picture}

\usepackage[utf8]{inputenc}
\usepackage{graphicx}
\usepackage{xcolor}
\usepackage{amsmath}
\usepackage{amsthm}
\usepackage{bm}
\usepackage{bbm}
\usepackage{comment}
\usepackage{appendix}
\usepackage{mathdots}
\usepackage{lipsum}
\usepackage{verbatim}
\usepackage{natbib}
\usepackage{nccmath}
\usepackage{amsfonts} 
\newcommand{\dya}[1]{\ket{#1}\!\bra{#1}}
\newcommand{\BC}{\mathcal{B}}
\newcommand{\CC}{\mathcal{C}}

\newcommand{\EC}{\mathcal{E}}

\newcommand{\TC}{\mathcal{T}}

\renewcommand{\leq}{\leqslant}

\renewcommand{\Im}{\text{Im}}

\renewcommand{\vec}[1]{\boldsymbol{#1}}  % Bold vectors instead of arrow vectors

\newcommand*{\id}{\openone}

\newcommand{\bs}{\textsf{BS}}

\newcommand{\sg}{\sigma }

\newcommand{\SWAP}{{\rm SWAP}}

\newcommand{\mcu}{\mathcal{U}}
\newcommand{\mcm}{\mathcal{M}}

\newcommand{\mcw}{\mathcal{W}}

\newcommand{\mbe}{\mathbb{E}}

\def\be{\begin{equation}}
\def\ee{\end{equation}}
\def\bs{\begin{split}}
\def\e{\end{split}}
\def\ba{\begin{eqnarray}}
\def\bea{\begin{eqnarray}}

\def\tea{\end{eqnarray}}
\def\ea{\end{eqnarray}}
\def\eea{\end{eqnarray}}

\def\d{\delta}

\def\d{\delta}

\definecolor{antonio}{rgb}{.2,.5,.1}

\newcommand\mbb[1]{\mathbb{#1}}
\newcommand\mf[1]{\mathfrak{#1}}

\newtheorem{theorem}{Theorem}
\newtheorem{lemma}{Lemma}
\newtheorem{result}{Result}
\usepackage{thm-restate}

\usepackage{amssymb}
\usepackage{dsfont}

\def\be{\begin{equation}}
\def\te{\end{equation}}
\def\ee{\end{equation}}
\def\ba{\begin{eqnarray}}
\def\bea{\begin{eqnarray}}

\def\tea{\end{eqnarray}}
\def\ea{\end{eqnarray}}
\def\eea{\end{eqnarray}}

\begin{document}

\title{Random ensembles of symplectic and unitary states are indistinguishable}

\author{Maxwell West}
\affiliation{School of Physics, University of Melbourne, Parkville, VIC 3010, Australia}
\affiliation{Theoretical Division, Los Alamos National Laboratory, Los Alamos, NM 87545, USA}

\author{Antonio Anna Mele}
\affiliation{Dahlem Center for Complex Quantum Systems, Freie Universität Berlin, 14195 Berlin, Germany}
\affiliation{Theoretical Division, Los Alamos National Laboratory, Los Alamos, NM 87545, USA}

\author{Mart\'{i}n Larocca}
\affiliation{Theoretical Division, Los Alamos National Laboratory, Los Alamos, NM 87545, USA}
\affiliation{Center for Non-Linear Studies, Los Alamos National Laboratory, 87545 NM, USA}

\author{M. Cerezo}
\thanks{cerezo@lanl.gov}
\affiliation{Information Sciences, Los Alamos National Laboratory, Los Alamos, NM 87545, USA}

\begin{abstract}
A unitary state $t$-design is an ensemble of pure quantum states whose moments match up to the $t$-th order those of states uniformly sampled from a $d$-dimensional Hilbert space. Typically, unitary state $t$-designs are obtained by evolving some reference pure state with unitaries from an ensemble that forms a design over the unitary group $\mathbb{U}(d)$, as unitary designs induce state designs. However, in this work we study whether Haar random symplectic states --i.e., states obtained by evolving some reference state with unitaries sampled according to the Haar measure over $\mathbb{SP}(d/2)$-- form unitary state $t$-designs. Importantly, we recall that  random symplectic unitaries fail to be unitary designs for $t>1$, and that, while it is known that symplectic unitaries are universal, this does not imply that their Haar measure leads to a state design. Notably, our main result states that Haar random symplectic states form unitary $t$-designs for all $t$, meaning that their distribution is unconditionally indistinguishable from that of unitary Haar random states, even with tests that use infinite copies of each state. As such, our work showcases the intriguing possibility of creating state $t$-designs using ensembles of unitaries which do not constitute designs over $\mathbb{U}(d)$ themselves, such as ensembles that form $t$-designs over $\mathbb{SP}(d/2)$.  
\end{abstract}

\maketitle

\section{Introduction}
Random quantum states play a ubiquitous role in quantum information sciences as their use is fundamental for randomized measurements protocols~\cite{renes2004symmetric,elben2022randomized,scott2006tight,huang2020predicting}, state discrimination~\cite{ambainis2007quantum,dankert2009exact,nielsen2002simple,dankert2009exact}, quantum sensing~\cite{smith2013quantum}, quantum supremacy experiments~\cite{boixo2018characterizing,arute2019quantum,wu2021strong,dalzell2022randomquantum,oszmaniec2022fermion,huang2021provably}; as well as to understand quantum chaos~\cite{cotler2023emergent,varikuti2024unraveling,roberts2017chaos,choi2023preparing,dowling2023scrambling}, information scrambling~\cite{styliaris2020information,hosur2016chaos}, and quantum machine learning schemes~\cite{cerezo2020variationalreview,larocca2024review}. As such, the quest for ensembles of states that exactly (or approximately) reproduce the statistics of Haar random states has received considerable attention~\cite{harrow2009random,brandao2016local,hunter2019unitary,haferkamp2022random,haferkamp2021improved,brown2010random,nakata2017efficient,Haferkamp2022randomquantum,harrow2018approximate,chen2024efficient,chen2024incompressibility,belkin2023approximate,mittal2023local,schuster2024random,deneris2024exact,iosue2024continuous,webb2016clifford,zhu2016clifford,cotler2023emergent,choi2023preparing}.

In this context, perhaps the simplest way of obtaining state designs is via unitary designs, i.e., sets of unitaries that match (up to a certain degree) the moments of Haar random unitaries over. This is due to the fact that unitary $t$-designs induce state $t$-design~\cite{mele2023introduction}. For instance, the Clifford group has been shown to form a $3$-design~\cite{webb2016clifford,zhu2016clifford}, meaning that an ensemble  of states obtained by evolving some reference state via uniformly sampled Clifford unitaries will match up to the third-moment the distribution of unitary Haar random states. Similarly, sending reference states through sufficiently deep random circuits~\cite{harrow2009random,brandao2016local,hunter2019unitary,haferkamp2022random,haferkamp2021improved,brown2010random,nakata2017efficient,Haferkamp2022randomquantum,harrow2018approximate,chen2024efficient,chen2024incompressibility,belkin2023approximate,mittal2023local,schuster2024random,deneris2024exact} will also produce unitary state $t$-designs as the circuits themselves form unitary designs.

The problem of finding unitary state designs becomes much more  intricate if one wishes to create them from an ensemble of unitaries that does not constitute a design over the unitary group. From a mathematical perspective this would imply that while the moments of the whole unitary matrices themselves do not match those of the unitary group, the moments of a single column do. The previous indicates that the correlations distinguishing one ensemble of unitaries from another are ultimately encoded in the correlations between columns, and not within them.  

In this work we  study whether Haar random symplectic unitaries lead to state designs (see Fig.~\ref{fig:schematic}). Importantly, it is well known that random symplectic unitaries fail to be $t$-designs for $t> 1$~\cite{collins2006integration,garcia2024architectures}. Then, while it has been shown that symplectic unitaries are universal~\cite{zimboras2015symmetry,oszmaniec2017universal}, in the sense that any pure state can be mapped into any other by a symplectic unitary, such a result does not imply that evolving some reference state via random symplectic unitaries will lead to a unitary state design. In this context, our theorem states that Haar random symplectic states form $t$-designs for all $t$. Indeed, our proof strategy relies on computing the moments for this ensemble of states and proving that $\forall t$ we recover a projector into the symmetric subspace of the $t$-th fold tensor product of the Hilbert space~\cite{harrow2013church}.  We also present several implications of these results, such as the fact that random ensembles of symplectic and unitary states are unconditionally indistinguishable. Ultimately,  we hope that our work will serve as blueprints to determining new ways to generate quantum state designs.

\section{Preliminaries}

Let $\HC$ be a $d$-dimensional Hilbert space (with $d$ even), and let $\mathbb{U}(d)$ denote the group of $d\times d$ unitary matrices acting on $\HC$. Then, let $\mathbb{SP}(d/2)\subseteq \mathbb{U}(d)$ be the unitary symplectic group, consisting of all $d\times d$ unitary matrices that satisfy 
\begin{equation}\label{eq:symplectic-unitaries}
    U^{T} \Omega U = \Omega\,,
\end{equation}
where $\Omega$ is a non-degenerate anti-symmetric bilinear such that
\begin{equation}\label{eq:props-Omega}
    \Omega^2=-\id_d\,,\quad \text{and} \quad\Omega\Omega^T=\Omega^T\Omega=\id_d\,.
\end{equation}
Importantly, $\Omega$ is not uniquely defined, and we here assume that $\Omega$ takes its canonical form 
\begin{equation} \label{eq:omega}
	\Omega=\begin{pmatrix} 0& \id_{d/2} \\ - \id_{d/2} & 0\end{pmatrix}\,,
\end{equation}
with $\id_{d/2}$ being the $d/2 \times d/2$ identity matrix. 

Then, let  $\EC=\{\ket{\psi}\}$ be an ensemble of pure quantum states obtained by sampling from $\HC$ according to some  probability distribution $d\psi_\EC$. We say that $\EC$ forms a  unitary state $t$-design if its moments match those of the Haar ensemble up to the $t$-th moment. Here, by Haar ensemble we mean the ensemble obtained by sampling from the uniform distribution $d\psi_{{\rm Haar}}$ over $\HC$, or equivalently, the ensemble obtained by evolving a reference state $\ket{\psi_0}$ with unitaries sampled according to the Haar measure $d\mu_{\mathbb{U}}$ over $\mathbb{U}(d)$. That is, $\EC$ is a $t$-design if
\begin{align}
\int_{\HC} d\psi_\EC \dya{\psi}^{\otimes t} &=\int_{\HC} d\psi_{{\rm Haar}} \dya{\psi}^{\otimes t}\nonumber\\ &=\int_{\mathbb{U}(d)} d\mu_{\mathbb{U}}\,  U^{\otimes t}\dya{\psi_0}^{\otimes t} (U\ad)^{\otimes t}\,.\nonumber
\end{align}
For simplicity of notation we will henceforth denote $\mathbb{E}_{\mathbb{U}}:=\int_{\mathbb{U}(d)} d\mu_{\mathbb{U}}$.
As mentioned above, unitary $t$-designs are sufficient to induce unitary state $t$-designs, but are they \textit{necessary}? For example, it is known that symplectic unitaries are transitive on the orbit of pure states~\cite{zimboras2015symmetry,oszmaniec2017universal}. Therefore, in the case of $t=1$, unitary symplectic designs are enough for unitary state designs. In what follows we will address the question: \textit{Is the ensemble of Haar random symplectic states  a unitary state $t$-design?} Here, by ensemble of Haar random symplectic states we refer to the one obtained by evolving some reference state  $\ket{\psi_0}$ with unitaries sampled according to the Haar measure $d\mu_{\mathbb{SP}}$ over $\mathbb{SP}(d/2)$ (see Fig.~\ref{fig:schematic}). That is,  we want to study whether 
\begin{align}
\mathbb{E}_{\mathbb{SP}}[U^{\otimes t}\dya{\psi_0}^{\otimes t} (U\ad)^{\otimes t}] &\overset{?}{=}\mathbb{E}_{\mathbb{U}}[ U^{\otimes t}\dya{\psi_0}^{\otimes t} (U\ad)^{\otimes t}]\,.\nonumber
\end{align}
Importantly, we note that the choice of reference state $\dya{\psi_0}^{\otimes t}$ is not important. This follows from the fact that the symplectic group is universal~\cite{zimboras2015symmetry,oszmaniec2017universal}, meaning that for any other choice  of reference state $\ket{\psi_0'}$ there exists a unitary $V\in \mathbb{SP}(d/2)$ such that (up to an unimportant global phase) $\ket{\psi_0}=V\ket{\psi_0'}$. Then, any such $V$ can be absorbed in the expectation value $\mathbb{E}_{\mathbb{SP}}$ due to the right- and left-invariance of the Haar measure.

\begin{figure}
    \centering
    \includegraphics[width=1\linewidth]{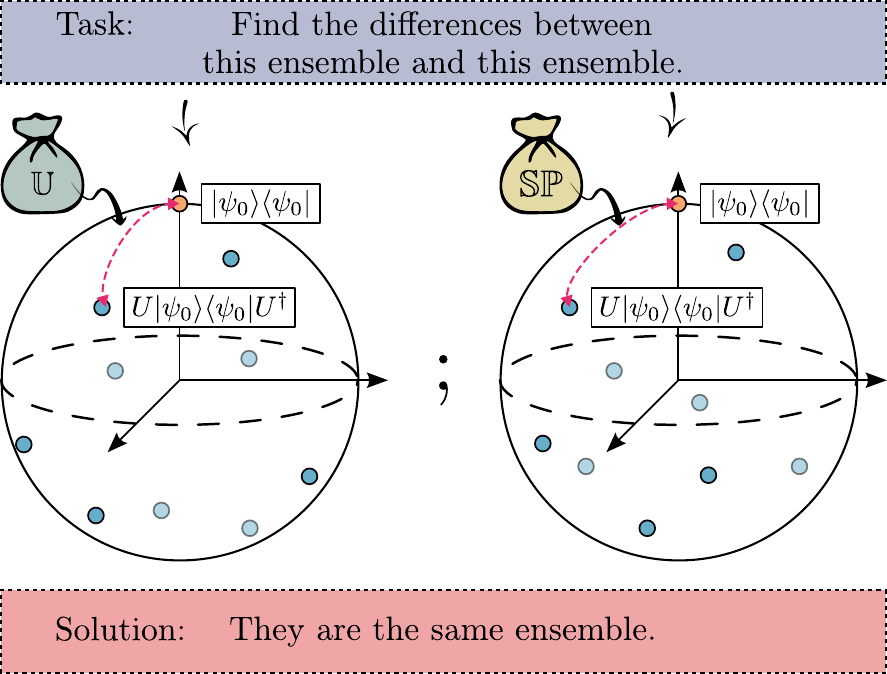}
    \caption{\textbf{Schematic representation of our results.} We show that the ensemble of unitary random states, i.e.,  the set of states obtained by evolving a reference state $\ket{\psi_0}$ with Haar random unitaries, is indistinguishable from that of symplectic random states.  }
    \label{fig:schematic}
\end{figure}

To answer the previous question we find it convenient to review a few concepts from the Weingarten calculus, a mathematical framework that allows for the computation of averages of tensors of unitaries sampled according to the Haar measure $d\mu$ over some group $G$. We refer the reader to~\cite{mele2023introduction} for additional details. In particular, we define the $t$-th order commutant of $G$ as the operator subspace
\begin{equation}
\CC^{(t)}(G)=\{A\in\BC(\HC^{\otimes t})\,|\,[A,U^{\otimes t}]=0\,, \forall U\in G\}\,,
\end{equation}
where $\BC(\HC^{\otimes t})$ denotes the set of bounded operators acting on the Hilbert space $\HC^{\otimes t}$. The importance of the $t$-th order commutant arises from the fact that  $\mathbb{E}_{G}[ U^{\otimes t}(\,\cdot\,) (U\ad)^{\otimes t}]$ projects $\BC(\HC\tt)$ into its subspace $\CC^{(t)}(G)$. Therefore, given a  $D$-dimensional basis $\{P_\mu\}_{\mu=1}^D$ of $\CC^{(t)}(G)$ and some operator $X\in\BC(\HC^{\otimes t})$, we can express 
\begin{equation}\label{eq:weingarten}
\mathbb{E}_{G}[ U^{\otimes t} X (U\ad)^{\otimes t}]=\sum_{\mu,\nu=1}^D (W^{-1})_{\mu\nu}\Tr[P_\mu X]P_\nu\,,
\end{equation}
where $W^{-1}$ is the so-called Weingarten matrix, obtained as the inverse (or pseudo-inverse) of the commutant's Gram matrix with entries $W_{\mu\nu}=\Tr[P_\mu  P_\nu]$. We refer the reader to Ref.~\cite{collins2006integration} for a discussion on the existence of this inverse (or pseudo-inverse) for the case when $G$ is the unitary or symplectic group.   We present a derivation for Eq.~\eqref{eq:weingarten} in  Appendix~\ref{ap:weingarten}.

For the case of $G=\mathbb{U}(d)$, we know from the Schur-Weyl duality that a basis for $\CC^{(t)}(\mathbb{U}(d))$ is given by the representation $P_d$ of the symmetric group $S_t$ that permutes the $d$-dimensional subsystems in the $t$-fold tensor product
Hilbert space, $\HC^{\otimes t}$. Here, $S_t$ is the group of bijections between the set $[t]$ and itself, i.e., the permutations of $t$ objects. As such, for any  $\pi\in S_t$, we have
\begin{equation}\label{eq:rep-S_k}
P_d(\pi)=\sum_{i_1,\dots,i_t=1}^{d} |i_{\sigma^{-1}(1)},\dots,i_{\sigma^{-1}(t)} \rangle\langle i_1,\dots,i_t|\,.
\end{equation}
As mentioned above, the $S_t$-representation $P_d$ spans the commutant of the $t$-fold standard action of the unitary group
\begin{equation}
\Im({P_d}) = \CC^{(t)}(\mbb{U}(d))\, \subset \BC(\HC^{\otimes t})\,.
\end{equation}
An important consequence of this result is that one can readily prove the following equality~\cite{harrow2013church}
\begin{align}
\mathbb{E}_{\mathbb{U}}[ U^{\otimes t}\dya{\psi_0}^{\otimes t} (U\ad)^{\otimes t}]&=\Pi_{{\rm sym}}^{(t)}\nonumber\\
&=\frac{\sum_{\pi \in S_t}P_d(\pi)}{d(d+1)\cdots(d+t-1)}\,,
\end{align}
where $\Pi_{{\rm sym}}^{(t)}$ is the normalized projector into the symmetric subspace of $\HC^{\otimes t}$.

Instead, when $G=\mathbb{SP}(d/2)$, its commutant is a representation $F_d$ of the Brauer algebra $\mathfrak{B}_t(-d)$~\cite{collins2006integration}.  
Let $M_t$ denote the set of all possible pairings of $2t$ object. For example, in the context of graph theory this corresponds to the set of all perfect matchings between the two bipartitions of the complete bipartite graph $K_{t,t}$. It can also be understood as the number of exact set covers of a universe set $U=[2t]$ with subsets of size two. Explicitly, any element $\sg \in M_t$ (a pairing) can be specified by $t$ disjoint pairs
\begin{equation}\label{eq:sigma-decomp}
    \sigma=\{\{\lambda_1, \sigma(\lambda_1)\}\cup\dots\cup\{\lambda_t, \sigma(\lambda_t)\}\}\,.
\end{equation}
Any pairing can be visualized as a diagram with two rows (say one at the top and another at the bottom), each containing $t$ elements, and where each pair in the pairing is depicted as a connection (edge) between two points. Edges can lie entirely within the top row, entirely within the bottom row, or between the top and bottom rows. Evidently, the symmetric group $S_t$ is a subset of $M_t$, that containing the pairings that strictly connect top and bottom rows.

We can define a multiplication of pairings $\sg,\pi$ in the following way: the diagrams of $\sg$ and $\pi$ are concatenated, and a scalar $\d$ is used to account for closed loops formed during the process. Evidently, for $\sg,\pi \in M_t$, $\sg \cdot \pi$ does not in general belong in $M_t$ (except at least one belongs in $S_t\subset M_t$), but belongs in the Brauer algebra, the algebra of finite \( \mathbb{Z}[\delta] \)-linear combinations of the elements in $M_t$, 
\[
B_t(\delta) = \left\{ \sum_{\sg \in M_t} c_\sg \sg \ \Big| \ c_\sg \in \mathbb{Z}[\delta] \right\}\,.
\]
Here, $\mbb{Z}[\d]$ denotes the ring of polynomials in the indeterminate $\d$ with integer coefficients. In other words, a general multiplication of two elements in $M_t$ leads to an element in the set times an integer power of $\delta$. 

As mentioned earlier, there is a natural representation $F_d$ of $B_t(-d)$ on $t$ copies of the $d$-dimensional quantum Hilbert space $\HC$ whose image corresponds to the commutant of the $t$-fold action of the unitary symplectic group
\begin{equation}
\Im({F_d}) = \CC^{(t)}(\mbb{SP}(d/2))\, \subset \BC(\HC^{\otimes t  })\,.
\end{equation}
The representation $F_d$ is explicitly given by (see Ref.~\cite{garcia2024architectures})
\small
\begin{align} 
F_d(\sigma) = \sum_{i_1,\dots,i_{2t}=1}^{d}\prod_{\gamma=1}^{t}  &\Omega_{{\sigma(\lambda_\gamma)}}^{h(\lambda_\gamma,\sigma(\lambda_\gamma))}\ket{i_{t+1},i_{t+2},\dots,i_{2t}} \label{eq:rep-b_t}\\ &\times\bra{i_1,i_2,\dots,i_t} \,\Omega_{{\sigma(\lambda_\gamma)}}^{h(\lambda_\gamma,\sigma(\lambda_\gamma))}\delta_{i_{\lambda_\gamma}, i_{\sigma(\lambda_\gamma)}} \,,\nonumber
\end{align} 
\normalsize
where $h(\lambda_\gamma,\sigma(\lambda_\gamma))=1$ if $\lambda_\gamma,\sigma(\lambda_\gamma)\leq n$ or if $\lambda_\gamma,\sigma(\lambda_\gamma)> n$ and zero otherwise, and where $\Omega_{\sigma(\lambda_\gamma)}$ indicates that the $\Omega$ matrix of Eq.~\eqref{eq:omega} acts on the  $\sigma(\lambda_\gamma)$-th copy of the Hilbert space. Clearly, if $\sigma$ is a permutation in $S_t\subset \mathfrak{B}_t(-d)$, then $P_d(\sigma)=F_d(\sigma)$. (In Fig.~\ref{fig:Brauer} in the appendices we explicitly present some elements of the Brauer algebra for $t=1,2,3$ as well as their schematic representation.)

\section{Main result}

Let $\mf{S}_t(\d)$ denote the subalgebra of $\mf{B}_t(\d)$ consisting of $\mbb{Z}(\d)$-weighted permutations, 
\begin{equation}
\mf{S}_t(\d) = \Big\{ \sum_{\sg \in S_t} c_\sg \sg \ \Big| \ c_\sg \in \mathbb{Z}[\delta] \Big\} \subset \mf{B}_t(\d)    \,.
\end{equation}
The following lemma plays a central role in our derivations.  

\begin{lemma}\label{lem:no-overlap}
Let $\ket{\psi_0}\in\HC$ be an arbitrary pure quantum state,  and consider $\sigma\in \mathfrak{B}_t(-d)\backslash \mf{S}_t(-d)$ (i.e., an element of the Brauer algebra which is not a permutation). Then, it follows that 
\begin{equation}\label{eq:permutation}
F_d(\sigma) \dya{\psi_0}^{\otimes t}=\dya{\psi_0}^{\otimes t}F_d(\sigma) =0\,,
\end{equation}
as well as 
\begin{equation}\label{eq:fpo}
 F_d(\sigma)\Pi_{{\rm sym}}^{(t)}=\Pi_{{\rm sym}}^{(t)}F_d(\sigma) =0\,.
\end{equation}
\end{lemma}
We refer the reader to the Appendix\footnote{ In Appendix~\ref{ap:rep} we also present a representation-theoretical interpretation of Lemma~\ref{lem:no-overlap}. } for the proof of Lemma~\ref{lem:no-overlap} as well as that of  our other results. The following theorem, our main result, is a direct consequence of  Lemma~\ref{lem:no-overlap}.
\begin{theorem}\label{theo:main}
Let $\ket{\psi_0}\in\HC$, then for all $t$ one has
\small
\begin{align}
\mathbb{E}_{\mathbb{SP}}[ U^{\otimes t}\dya{\psi_0}^{\otimes t} (U\ad)^{\otimes t}]&=\frac{\sum_{\pi \in S_t}P_d(\pi)}{d(d+1)\cdots(d+t-1)}\\
&=\mathbb{E}_{\mathbb{U}}[ U^{\otimes t}\dya{\psi_0}^{\otimes t} (U\ad)^{\otimes t}]\,,\nonumber
\end{align}
\normalsize
meaning that symplectic random states form unitary state $t$-designs for all $t$.
\label{th:main}
\end{theorem}

It follows from Theorem~\ref{theo:main} that ensembles of unitary and symplectic Haar random states are statistically indistinguishable, even if one is allowed to query an arbitrary number of copies of each state.  As further discussed in Appendix~\ref{ap:indist} we can formalize this statement as:

\begin{restatable}{result}{indist}
    Consider a quantum state $|\psi\rangle$ that is sampled with probability $1/2$ from the distribution of unitary random states and with probability $1/2$ from the distribution of symplectic random states. No quantum experiment can distinguish whether $|\psi\rangle$ was sampled from the unitary or symplectic ensemble with a success probability $>1/2$, even if $|\psi\rangle$ is queried an arbitrarily large number of times.
\end{restatable}

At this point one may wonder if the results in Theorem~\ref{theo:main} extend beyond pure states, i.e., whether the ensembles obtained by evolving some reference mixed state with Haar random unitaries sampled from $\mathbb{U}(d)$ and $\mathbb{SP}(d/2)$  have the same moments. However, we can prove that this is not the case, as the following theorem holds.  
\begin{theorem}\label{theo:rank-two}
    There exists a rank-two quantum state $\rho$ for which
    \begin{equation}
        \mathbb{E}_{\mathbb{U}}[ U^{\otimes 2}\rho^{\otimes 2}(U\ad)^{\otimes 2}]\neq \mathbb{E}_{\mathbb{SP}}[ U^{\otimes 2}\rho^{\otimes 2}(U\ad)^{\otimes 2}]\,.
    \end{equation}
\end{theorem}

The crucial difference between pure and mixed states is that a result such as that in Eq.~\eqref{eq:permutation} does not hold for states with rank higher than one. That is, given some $\rho=\lambda_0\dya{\psi_0}+\lambda_1\dya{\psi_1}$ (with $\lambda_0+\lambda_1=1$ and $|\langle\psi_0|\psi_1\rangle|^2=0$) one generally has that $F_d(\sigma) \rho^{\otimes t}\neq \rho^{\otimes t}F_d(\sigma) \neq 0$ for $\sigma\in\mathfrak{B}_t(-d)\backslash \mf{S}_t(-d)$.

\section{Applications and open questions}

In this section we review a few  potential applications of our results, as well as new research directions  worth highlighting. We note that the list of topics covered here is not meant to be comprehensive, and that an in-depth exploration of each application and open question goes beyond the scope of this work. As such, we hope that these will inspire the community to further  explore  the implications of Theorem~\ref{theo:main}. As we will see, the common theme of our potential applications is to consider tasks where random unitaries are applied to some quantum state, and then to replace the unitaries from $\mathbb{U}(d)$ to unitaries sampled from $\mathbb{SP}(d/2)$.

\subsection{Approximate state design from shorter depth unitaries}

One of the simplest ways to generate approximate state $t$-designs is to send states through circuits that form approximate unitary $t$-designs~\cite{harrow2009random,brandao2016local,hunter2019unitary,haferkamp2022random,haferkamp2021improved,brown2010random,nakata2017efficient,Haferkamp2022randomquantum,harrow2018approximate,chen2024efficient,chen2024incompressibility,belkin2023approximate,mittal2023local,schuster2024random,deneris2024exact,iosue2024continuous,webb2016clifford,zhu2016clifford,cotler2023emergent,choi2023preparing}. Hence,  one can instead send the states through a circuit which forms an approximate symplectic $t$-design. 

For instance, a typical paradigm to form approximate unitary $t$-designs is via random circuits where local two-qubit gates --which  are independent and identically sampled from $\mathbb{U}(4)$-- act on neighboring pairs of qubits in a brick-like fashion and over some chosen topology. Crucially, if enough layers of this architecture are used, the circuit will form an approximate $t$-designs over $\mathbb{U}(d)$~\cite{harrow2009random,brandao2016local,hunter2019unitary,haferkamp2022random,haferkamp2021improved,brown2010random,nakata2017efficient,Haferkamp2022randomquantum,harrow2018approximate,chen2024efficient,chen2024incompressibility,belkin2023approximate,mittal2023local,schuster2024random,deneris2024exact,iosue2024continuous,webb2016clifford,zhu2016clifford,cotler2023emergent,choi2023preparing,ragone2023unified}. Here, one could instead use local gates whose distribution converges to the Haar over $\mathbb{SP}(d/2)$ instead. A direct calculation using the results in~\cite{garcia2024architectures,deneris2024exact}, explicitly shown in  Appendix~\ref{ap:brick}, leads to the following result:
\begin{result}\label{res:params}
A random circuit in a one-dimensional topology composed of symplectic random two-qubit gates can form an $\epsilon$-approximate  $2$-design with $~60\%$ less parameters than a circuit with the same topology and with unitary random two-qubit gates. 
\end{result}

We expect that symplectic circuits could also  improve more advanced techniques to form approximate designs, such as those recently introduced in Ref.~\cite{schuster2024random}. 

\subsection{Symplectic classical shadows}

As a second application of our results, we demonstrate that any classical shadows protocol~\cite{huang2020predicting,zhao2021fermionic,van2022hardware,bertoni2022shallow,king2024triply,jerbi2023shadows,koh2022classical,chen2021robust,sauvage2024classical,wan2022matchgate,bermejo2024quantum,angrisani2024classically} based on sampling random unitaries from the unitary group is equivalent to a ``symplectic classical shadows'' protocol where one samples random unitaries from the symplectic group.  That is, 
\begin{restatable}{result}{shad}\label{result:shad}
Symplectic classical shadows are equivalent to unitary classical shadows.
\end{restatable}

Interestingly, this result holds true independently of the state from which the shadows are taken (i.e., the target state can be pure or mixed). The equivalence between symplectic and unitary shadows follows by considering the shadow protocol in  the Heisenberg picture. Specifically, and as we discuss in Appendix~\ref{ap:shad}, Result~\ref{result:shad} follows from the fact that a classical shadows protocol -- which consists of a choice of a unitary ensemble $\mathcal{U}$ and a measurement basis $\{\ket{w}\}_w$ -- is completely determined by the values of the expectations
\begin{equation}
\mathbb{E}_\mcu\left[U ^{\otimes t}\left(\sum_w\ketbra{w}^{\otimes t}\right) U ^{\dagger \otimes t}\right],\quad \text{for} \quad t=2,3\,.
\label{eq:shadowsmoments}
\end{equation}
Thus, as we can see in Eq.~\eqref{eq:shadowsmoments} the properties of the shadow protocol are determined by the expectation value over the (pure) measurement basis. 
In particular, any  unitary ensemble which induces a state 3-design with respect to any computational basis reference state leads to a classical shadow protocol identical to that of the full unitary group.  For example, this was exploited in Ref.~\cite{huang2020predicting} to conclude that sampling random Clifford gates, forming as they do a unitary state 3-design (in fact, a unitary 3-design), is sufficient to effect unitary shadows. Having seen in Theorem~\ref{theo:main} that the symplectic random states form unitary state $t$-design for all $t$, it  follows that symplectic classical shadows are likewise equivalent to their unitary counterpart.

\subsection{Symplectic state tomography}

It is known~\cite{guta2018faststatetomographyoptimal,Lowe2021} that any ensemble forming a unitary state $2$-design (or any unitary ensemble for which the expression in Eq.~\eqref{eq:shadowsmoments} yields, for $t=2$, the same result as the Haar unitary ensemble) leads to a full state tomography protocol using only single-copy measurements, which is almost optimal in terms of sample complexity~\cite{chen2023doesadaptivityhelpquantum} among all protocols restricted to single-copy measurements. The term ``almost'' refers to the fact that the sample complexity bound matches the information-theoretic lower bound~\cite{chen2023doesadaptivityhelpquantum} for single-copy possibly adaptive tomography  protocols, up to logarithmic factors in the Hilbert space dimension $d$.
Specifically, it has been shown that performing (state $2$-design) randomized single-copies measurements over $\tilde{O}(d^3/\varepsilon^2)$ copies of an unknown state $\rho$ allows one to construct an estimator $\hat{\rho}$ that, with high probability, is $\varepsilon$-close to $\rho$ in trace norm (where $\tilde{O}$ hides $\log(d)$ factors).
Thus, based on our Theorem~\ref{th:main}, we have that:
\begin{result}\label{crllr:tomo}
There exists a full state tomography protocol based on random symplectic $2$-design unitaries followed by computational basis measurements on single copies of the unknown state, which achieves provably optimal sample complexity (up to logarithmic factors in the Hilbert space dimension) within the class of single-copy measurement protocols.
\end{result}
We note that this protocol~\cite{guta2018faststatetomographyoptimal,Lowe2021} is closely related to the classical shadow protocol~\cite{huang2020predicting}, as it employs the same estimator for the unknown state.

\subsection{Open question: Can we find efficient  symplectic $t$-designs?}
One of the most intriguing implications of our results is that any ensemble of unitaries which forms a symplectic $t$-design can be used  in lieu of Haar random symplectic unitaries to form a unitary state $t$-design.  This  raises the immediate question of whether we can find sets of efficiently implementable unitaries which form a design over the symplectic group. For example, we pose the question: 

\medskip
\textit{Do the unitaries in ${\rm Cl}(n)\cap \mathbb{SP}(2^n/2)$ form a $3$-design over $\mathbb{SP}(2^n/2)$?}
\medskip

Here, ${\rm Cl}(n)$ denotes the $n$-qubit Clifford group.
Such a question is motivated by the fact that it is known that intersecting the Clifford with other subgroups of the unitary group leads to $3$-designs~\cite{hashagen2018real,nebe2001invariants,nebe2006self,wan2022matchgate,Mitsuhashi_2023}.

\medskip

\section{Discussion}

In this work we show that ensembles of symplectic random states are indistinguishable from ensembles of unitary random states, in the sense that their moments match to all orders.  Despite the relative simplicity of our main theorem, as well as of its associated proof, Theorem~\ref{theo:main} has important conceptual and practical implications.

On the conceptual side, the fact that unitary state designs can be formed by ensembles of unitaries which do not form designs over the unitary group is quite intriguing. Indeed, one may wonder if there exists other ensembles beyond Haar random symplectic unitaries from which unitary state designs can be found. More generally, our work also paves the way to study state designs with respect to other groups. That is, if we are interested in forming a state design over a group $G$, are there subgroups of $G$ which lead to $G$ state designs (even if they are not designs over $G$)? 

Then, on the practical side, our work has direct implications for quantum information protocols that require unitary state $t$-designs. Indeed, instead of evolving some reference state by   random unitaries from $\mathbb{U}(d)$, one could instead evolve them by random unitaries from $\mathbb{SP}(d/2)$ or even by unitaries from some ensemble which forms a $t$-design over $\mathbb{SP}(d/2)$. Ultimately, the goal is to find smaller ensembles to sample from, as we can expect that their implementations will be simpler. In this context,  we have left open as an important question whether we can find efficiently implementable symplectic Clifford circuits which form $3$-designs over $\mathbb{SP}(d/2)$. Indeed, if they do, symplectic Cliffords could become the standard for implementing up-to-three unitary state designs.

\section{Acknowledgments}

The authors thank Diego Garc\'ia-Mart\'in, Manuel G. Algaba, and Yigit Subasi for insightful discussions. MW acknowledges the support of the Australian government research training program scholarship and the IBM Quantum Hub at the University of Melbourne. AAM acknowledges support by the German Federal Ministry for Education and Research (BMBF) under the project FermiQP. MW and AAM were supported by the U.S. DOE through a quantum computing program sponsored by the Los Alamos National Laboratory (LANL) Information Science \& Technology Institute. ML and MC acknowledge support by the Laboratory Directed Research and Development (LDRD) program of LANL under project numbers 20230049DR and 20230527ECR. ML was also supported by  the Center for Nonlinear Studies at LANL. MC also acknowledges initial support by from LANL ASC Beyond Moore’s Law project.

\bibliography{quantum,ref}

\onecolumngrid
\appendix
\setcounter{lemma}{0}
\setcounter{theorem}{0}
\setcounter{remark}{0}

\section{Weingarten calculus}\label{ap:weingarten}

As discussed in the main text, the twirl  $\mathbb{E}_{G}[ U^{\otimes t}(\cdot) (U\ad)^{\otimes t}]$ is a projection into $\CC^{(t)}(G)$, meaning that  given a  $D$-dimensional basis $\BC=\{P_\nu\}_{\nu=1}^D$ of $\CC^{(t)}(G)$ we have 
\begin{equation}\label{eq:twirled_X_comm}
    \mathbb{E}_{G}[ U^{\otimes t}(X) (U\ad)^{\otimes t}]=\sum_{\nu=1}^{D} c_\nu(X) P_\nu\,, \quad \text{with} \quad P_\mu\in\BC\,.
\end{equation}
To solve Eq.~\eqref{eq:twirled_X_comm}, we needs to find the $D$ unknown coefficients $\{c_\nu(X)\}_{\nu=1}^D$, which can be accomplished by determining a set of $D$ equations  and solving the resulting linear system problem. In particular, multiplying both sides of Eq.~\eqref{eq:twirled_X_comm} by some $P_\mu\in \BC$ leads to 
\begin{align}
   P_\mu  \mathbb{E}_{G}[ U^{\otimes t}(X) (U\ad)^{\otimes t}]&=\sum_{\nu=1}^{D} c_\nu(X) P_\mu P_\nu\nonumber \\
 \rightarrow\quad \mathbb{E}_{G}[ U^{\otimes t}(P_\mu X) (U\ad)^{\otimes t}]&=\sum_{\nu=1}^{D} c_\nu(X) P_\mu P_\nu\,,\label{eq:twirled_X_comm_2}
\end{align}
where in the second line we have used the fact that $P_\mu$ belongs to $\CC^{(t)}(G)$.
Then, if we take the trace of both sides of Eq.~\eqref{eq:twirled_X_comm_2} we obtain
\begin{align}\label{eq:LSP}
   \Tr[P_\mu X ] =\sum_{\nu=1}^{D}  \Tr[P_\mu P_\nu ] c_\nu(X)\,,
\end{align}
where we used the fact that $\Tr[ \mathbb{E}_{G}[ U^{\otimes t}(X) (U\ad)^{\otimes t}]]=\Tr[X]$. Repeating Eq.~\eqref{eq:LSP} for all $P_\mu$ in $ \BC$ leads to $D$ equations. Thus, we can find the vector of unknown coefficients $\vec{c}(X)=(c_1(X),\ldots,c_D(X))$ by solving
\begin{equation}
    W\cdot \vec{c}(X)=\vec{b}(X)\,,
    \label{eq:lins}
\end{equation} 
where $\vec{b}(X)=(\Tr[X P_1],\ldots, \Tr[X P_D])$. Here,  $W$ is a $D\times D$  Gram matrix of the commutant whose entries are $(W)_{\mu\nu}=\Tr[P_\mu P_\nu]$. We can then find that a valid solution to the previous Eq.\eqref{eq:lins} is given by
\begin{equation}\label{eq:inverse-vec-c}
    \vec{c}(X)=W^{-1}\cdot\vec{b}(X)\,,
\end{equation}
where $W^{-1}$ is the pseudo-inverse of $W$, and so recover Eq.~\eqref{eq:weingarten}.

\section{Proof of Lemma~\ref{lem:no-overlap}}

Here we provide a proof for Lemma~\ref{lem:no-overlap}, which we restate for convenience.

\begin{lemma}\label{lem:no-overlap-SM}
Let $\ket{\psi_0}\in\HC$ be an arbitrary quantum state,  and let $\sigma$ be  such that $\sigma\in \mathfrak{B}_t(-d)\backslash \mf{S}_t(-d)$ (i.e., an element of the Brauer algebra which is not a permutation). Then, it follows that 
\begin{equation}
F_d(\sigma) \dya{\psi_0}^{\otimes t}=\dya{\psi_0}^{\otimes t}F_d(\sigma) =0\,,
\end{equation}
as well as 
\begin{equation}\label{eq:fpo-SM}
 F_d(\sigma)\Pi_{{\rm sym}}^{(t)}=\Pi_{{\rm sym}}^{(t)}F_d(\sigma) =0\,.
\end{equation}
\end{lemma}

Before proceeding to the proof, let us recall from the main text that all the elements $\sigma\in\mathfrak{B}_t(\delta)$ can  be completely specified by $t$ disjoint pairs, as 
\begin{equation}\label{eq:sigma-decomp-SM}
    \sigma=\{\{\lambda_1, \sigma(\lambda_1)\}\cup\dots\cup\{\lambda_t, \sigma(\lambda_t)\}\}\,.
\end{equation}
Then, the representation $F_d$ is defined for an element $\sigma\in \mathfrak{B}_t(-d)$ as 
\begin{align} 
F_d(\sigma) = \sum_{i_1,\dots,i_{2t}=1}^{d}\prod_{\gamma=1}^{t}  \Omega_{{\sigma(\lambda_\gamma)}}^{h(\lambda_\gamma,\sigma(\lambda_\gamma))}\ket{i_{t+1},i_{t+2},\dots,i_{2t}} \times\bra{i_1,i_2,\dots,i_t} \,\Omega_{{\sigma(\lambda_\gamma)}}^{h(\lambda_\gamma,\sigma(\lambda_\gamma))}\delta_{i_{\lambda_\gamma}, i_{\sigma(\lambda_\gamma)}} \,,\label{eq:rep-b_t-SM}
\end{align} 
where $h(\lambda_\gamma,\sigma(\lambda_\gamma))=1$ if $\lambda_\gamma,\sigma(\lambda_\gamma)\leq n$ or if $\lambda_\gamma,\sigma(\lambda_\gamma)> n$ and zero otherwise, and where $\Omega_{\sigma(\lambda_\gamma)}$ indicates that the $\Omega$ matrix of Eq.~\eqref{eq:omega} acts on the  $\sigma(\lambda_\gamma)$-th copy of the Hilbert space. We refer the reader to Fig.~\ref{fig:Brauer} where we schematically represent elements of $\mathfrak{B}_t(-d)$ for $t=1,2,3$.

\begin{figure*}
    \centering
    \includegraphics[width=1\linewidth]{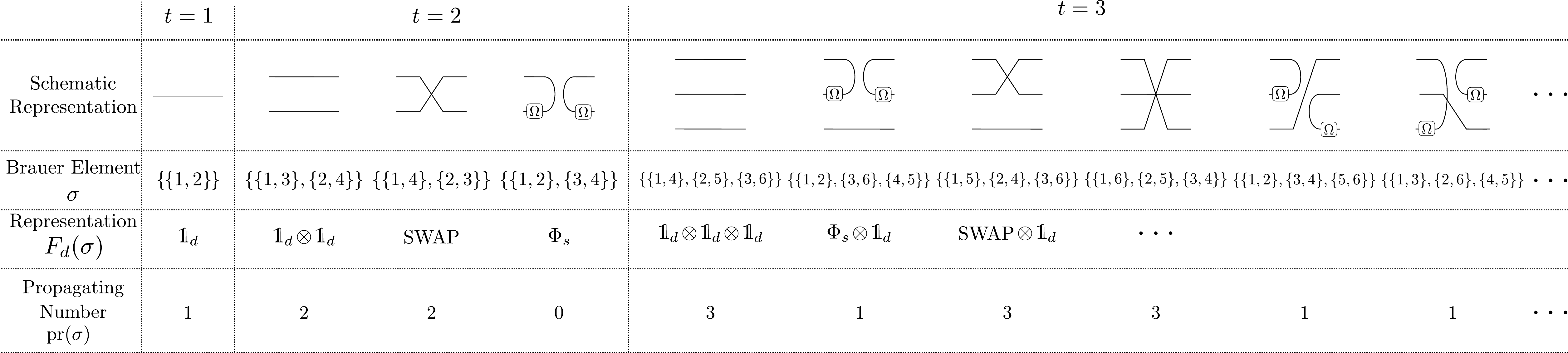}
    \caption{\textbf{Brauer algebra.} We show the schematic representation  of all elements of the Brauer algebra $\mathfrak{B}_t(-d)$ for $t=1,2$ and for some elements of $t=3$. We defined $\id_d=\sum_{i,j=1}^d\ket{ij}\bra{ij}$ as the $d\times d$ identity matrix, $\SWAP=\sum_{i,j=1}^d\ket{ij}\bra{ji}$ as the SWAP operator, and  $\Phi_s=\sum_{i,j=1}^d(\id_d\otimes \Omega)\ket{jj}\bra{ii}(\id_d\otimes \Omega)$. Here we also report the propagating number ${\rm pr}$ for each depicted element, were we recall that the propagating number is defined as the number of ``legs'' that cross from left to right~\cite{rubey2015combinatorics} and which take values in $\{t,t-2,t-4,\ldots\}$.}
    \label{fig:Brauer}
\end{figure*}

\begin{proof}
We start by showing that the representation of any element of the Brauer which  is not a permutation in the symmetric group, is antisymmetric with respect to at least one exchange of indexes. As such, take some $\sigma\in \mathfrak{B}_t(-d)\backslash \mf{S}_t(-d)$. That is, an element of the Brauer algebra that is not a permutation. By definition, it must follow that there exists at least one pair of indexes $\{\lambda_\eta,\sigma(\lambda_\eta)\}$ for which $\lambda_\gamma,\sigma(\lambda_\gamma)\leq n$. In diagrammatic tensor notation this means that there must be at least one ``loop'' on the left (see for instance third elements in the $t=1$ panel of Fig~\ref{fig:Brauer}, or the fifth element depicted for $t=2$). For instance, if we assume that the first pair in $\sigma$ as in Eq.~\eqref{eq:sigma-decomp} is $\{t+1,t+2\}$, then we can factorize $F_d(\sigma)$ as 
\begin{align} 
F_d(\sigma) =& \left(\sum_{j=1}^d\id_d\otimes \Omega|jj\rangle\right) \sum_{i_3,\dots}\prod_{\gamma=2}^{t}  \Omega_{{\sigma(\lambda_\gamma)}}^{h(\lambda_\gamma,\sigma(\lambda_\gamma))}\ket{i_{t+3},\dots,i_{2t}} \bra{i_1,i_2,\dots,i_t} \,\Omega_{{\sigma(\lambda_\gamma)}}^{h(\lambda_\gamma,\sigma(\lambda_\gamma))}\delta_{i_{\lambda_\gamma}, i_{\sigma(\lambda_\gamma)}} \,,
\end{align} 
which reveals that there is a loop on the first two indexes of the left (see Fig.~\ref{fig:Brauer}).

Then, consider the following chain of equalities:
\begin{align}
\SWAP \left(\sum_{j=1}^d\id_d\otimes \Omega|jj\rangle\right)&=\sum_{j=1}^d \Omega\otimes \id_d|jj\rangle
=\sum_{j=1}^d \id_d\otimes \Omega^T|jj\rangle=-\sum_{j=1}^d\id_d\otimes \Omega|jj\rangle\,,\label{eq:antisymmetric}
\end{align}
where we have defined $\SWAP=\sum_{i_1i_2}\ket{i_2i_1}\bra{i_1i_2}$. In the second line above, we have used the fact that the transpose trick holds, i.e., that for any $A,B\in\BC(\HC)$ then 
\begin{equation}
\sum_{j=1}^d A\otimes B|jj\rangle=\sum_{j=1}^d AB^T\otimes \id_d|jj\rangle=\sum_{j=1}^d\id_d\otimes BA^T|jj\rangle\,.\nonumber
\end{equation}
 Finally, in the last equality of Eq.~\eqref{eq:antisymmetric} we have simply used that $\Omega^T=-\Omega$ as detailed in Eq.~\eqref{eq:props-Omega}. As such, we have proven that the left- or right- multiplication of $F_d(\sigma)$ by any permutation that transposes the indexes of a loop leads to a minus sign.

Then, note that
\begin{equation}
\dya{\psi_0}^{\otimes t}=P_d(\pi)\dya{\psi_0}^{\otimes t}=\dya{\psi_0}^{\otimes t}P_d(\pi)\,,
\end{equation}
for any $\pi\in S_t$. From the previous, given any $\sigma$ in the Brauer algebra that is not a permutation, then we can take $\pi(\sigma)$ to be an element of $S_t$ which transposes the indexes of a loop in $\sigma$. Hence, we have 
\begin{align}
F_d(\sigma) \dya{\psi_0}^{\otimes t}&=F_d(\sigma) P_d(\pi(\sigma))\dya{\psi_0}^{\otimes t}=-F_d(\sigma) \dya{\psi_0}^{\otimes t}\,,
\end{align}
which proves that $F_d(\sigma) \dya{\psi_0}^{\otimes t}=0$. Similarly, we can start from $ \dya{\psi_0}^{\otimes t}F_d(\sigma)$  and take $\pi(\sigma)$ a permutation which transposes the indexes of a loop in $\sigma$ to the left, which leads to $ \dya{\psi_0}^{\otimes t}F_d(\sigma)=0$, establishing the first claim. 

Next, let us note that for any $\pi'\in S_t$
\begin{equation}\label{eq:inv-St}
    \sum_{\pi\in S_t} P_d(\pi)=P_d(\pi')\left(\sum_{\pi\in S_t} P_d(\pi)\right)=\left(\sum_{\pi\in S_t} P_d(\pi)\right)P_d(\pi')\,,
\end{equation}
which follows from the fact that $S_t$ forms a group and that $P_d$ is a group homomorphism, so that
\begin{equation}
    \left(\sum_{\pi\in S_t} P_d(\pi)\right)P_d(\pi')=\sum_{\pi\in S_t} P_d(\pi\cdot \pi')= \sum_{\tilde{\pi}\in S_t} P_d(\tilde{\pi})
\end{equation}
where we redefined the index in the summation as $\pi=\tilde{\pi}\cdot (\pi')^{-1}$. A similar derivation can be used when multiplying by $P_d(\pi')$ on the left. Then, similarly to the previous case, we have that for any $\sigma\in \mathfrak{B}_t(-d)\backslash \mf{S}_t(-d)$
\begin{align}
F_d(\sigma)\left(\sum_{\tilde{\pi}\in S_t} P_d(\tilde{\pi})\right) =\sum_{\pi\in S_t} F_d(\sigma)P_d(\pi'(\sigma)\circ \pi )  
=\sum_{\pi\in S_t} F_d(\sigma)P_d(\pi'(\sigma))P_d(\pi) =-\sum_{\pi\in S_t}  F_d(\sigma)P_d(\pi)\,.\label{eq:suminv}
\end{align}
The previous implies that $F_d(\sigma)\left(\sum_{\pi\in S_t} P_d(\pi)\right)=0$. A similar derivation can be used to show that $\left(\sum_{\pi\in S_t} P_d(\pi)\right)F_d(\sigma)=0$.
\end{proof}

\section{Proof of Theorem~\ref{theo:main}}

We begin by recalling Theorem~\ref{theo:main}.

\begin{theorem}\label{theo:main-SM}
Let $\ket{\psi_0}\in\HC$, then
\begin{align}\label{eq:twirl-brauer}
\mathbb{E}_{\mathbb{SP}}[ U^{\otimes t}\dya{\psi_0}^{\otimes t} (U\ad)^{\otimes t}]&=\frac{\sum_{\pi \in S_t}P_d(\pi)}{d(d+1)\cdots(d+t-1)}\nonumber\\
&=\mathbb{E}_{\mathbb{U}}[ U^{\otimes t}\dya{\psi_0}^{\otimes t} (U\ad)^{\otimes t}]\,.
\end{align}
\end{theorem}

And we also recall the definition of the normalized projector onto the symmetric subspace of $\HC^{\otimes t}$
\begin{equation}
   \Pi_{{\rm sym}}^{(t)}= \frac{\sum_{\pi \in S_t}P_d(\pi)}{d(d+1)\cdots(d+t-1)}\,.
\end{equation}

\begin{proof}
    Let us now recall from  Appendix~\ref{ap:weingarten} the main key step  behind the Weingarten calculus. In particular for the case of averages over $\mathbb{SP}(d/2)$ we  have
\begin{equation}
    \mathbb{E}_{G}[ U^{\otimes t}(X) (U\ad)^{\otimes t}]=\sum_{\sigma\in\mathfrak{B}_t(-d)} c_\sigma(X) F_d(\sigma)\,.
\end{equation}
Then, for the special case of $X=\dya{\psi_0}^{\otimes t}$, when building the linear system of equations to solve, we would find from Lemma~\ref{lem:no-overlap-SM} that all the entries in the $\vec{b}(\dya{\psi_0}^{\otimes t})$ vector corresponding to elements of the Brauer algebra that are not permutations are zero. Hence, re-ordering the basis of $\mathfrak{B}_t(-d)$  so that the first $t!$ elements are permutations and leveraging Lemma~\ref{lem:no-overlap-SM} along with the fact that $\Tr[\dya{\psi_0}^{\otimes t}P_{d}(\pi)]=1$ for all $\pi\in S_t$, we find that  
\begin{equation}\label{eq:c-block}
    \vec{b}(\dya{\psi_0}^{\otimes t})=(\vec{1}_{S},\vec{0})\,,\quad \text{and we write} \quad \quad \vec{c}(\dya{\psi_0}^{\otimes t})=(\vec{c}_{S}(\dya{\psi_0}^{\otimes t}),\vec{c}_{B/S}(\dya{\psi_0}^{\otimes t}))\,.
\end{equation}
Above, $\vec{1}_{S}$ is a vector of all ones of length $t!$, and  we also divided $\vec{c}(\dya{\psi_0}^{\otimes t})$ into the coefficients corresponding to permutations and to elements of the Brauer algebra which are not permutations.

Then, let us note that for any $\pi'\in S_t$, we can use Eq.~\eqref{eq:inv-St} to show that 
\begin{equation}
    \Tr[\Pi_{{\rm sym}}^{(t)} P_d(\pi') ]=1\,.
\end{equation}
Moreover, since by definition
\begin{equation}
    \mathbb{E}_{\mathbb{SP}}\left[ U^{\otimes t}\Pi_{{\rm sym}}^{(t)}(U\ad)^{\otimes t}\right]=\Pi_{{\rm sym}}^{(t)}\,,
\end{equation}
we find that  $\vec{c}\left(\Pi_{{\rm sym}}^{(t)}\right)=(\vec{1}_{S},\vec{0})$ and $\vec{b}\left(\Pi_{{\rm sym}}^{(t)}\right)=(\vec{1}_{S},\vec{0})$, which implies
\begin{equation}
    W\cdot (\vec{1}_{S},\vec{0})=(\vec{1}_{S},\vec{0})\,.
    \label{eq:eig1}
\end{equation} 
This means that $(\vec{1}_{S}, \vec{0})$ is an eigenvector of $W$ with eigenvalue 1.
Recall that $W$ is Hermitian, being a Gram matrix, and therefore admits an orthonormal eigendecomposition:
\begin{equation}
W = \sum_j a_j \ketbra{v_j}{v_j}\,,
\end{equation}
where $\{a_j\}_j$ are the eigenvalues, and $\{\ket{v_j}\}_j$ are the corresponding eigenvectors. While $W$ is not necessarily invertible, we can use its pseudoinverse, given by
\begin{equation}
W^{-1} = \sum_{j : a_j \neq 0} a_j^{-1} \ketbra{v_j}{v_j}\,.
\end{equation}
Multiplying Eq.~\eqref{eq:eig1} by $W^{-1}$ on both sides yields
\begin{equation}
    (\vec{1}_{S}, \vec{0}) = W^{-1} (\vec{1}_{S}, \vec{0})\,.
\end{equation}
Thus, we conclude that 
\begin{align}
    \vec{c}(\dya{\psi_0}^{\otimes t}) =W^{-1}\vec{b}(\dya{\psi_0}^{\otimes t}) = W^{-1}(\vec{1}_{S},\vec{0}) = (\vec{1}_{S}, \vec{0}),
\end{align}
and hence
\begin{equation}
    \mathbb{E}_{\mathbb{SP}}[ U^{\otimes t} \dya{\psi_0}^{\otimes t} (U\ad)^{\otimes t}] = \Pi_{{\rm sym}}^{(t)} = \frac{\sum_{\pi \in S_t} P_d(\pi)}{d(d+1)\cdots(d+t-1)}\,.
\end{equation}

Here we note that we can also recover this result with the following proof strategy. Given that  $\mathbb{E}_{\mathbb{SP}}[ U^{\otimes t}(\cdot) U^{\dagger \otimes t}]$ projects into its commutant, we know that there exist coefficients $\{c_\pi\}_\pi \cup \{c_\sigma\}_\sigma$ such that 
\begin{equation}
\mathbb{E}_{\mathbb{SP}}[ U^{\otimes t}\dya{\psi_0}^{\otimes t} (U\ad)^{\otimes t}]= \sum_{\pi\in S_t}c_{\pi}P_d(\pi) +  \sum_{\sigma\in \mathfrak{B}_t(-d)\backslash \mf{S}_t(-d)}c_{\sigma}F_d(\sigma)\,.
\end{equation}
By Lemma~\ref{lem:no-overlap-SM} we further know that each $F_d(\sigma)$ is annihilated by left- or right-multiplication by $\Pi_{{\rm sym}}^{(t)}$. Hence,  acting on both sides of the above equation with $\Pi_{{\rm sym}}^{(t)}$ leads to 
\begin{align*}
0&=\Pi_{{\rm sym}}^{(t)} \sum_{\sigma\in \mathfrak{B}_t(-d)\backslash \mf{S}_t(-d)}c_{\sigma}F_d(\sigma)\\
&=\Pi_{{\rm sym}}^{(t)} \left( \mathbb{E}_{\mathbb{SP}}[ U^{\otimes t}\dya{\psi_0}^{\otimes t} (U\ad)^{\otimes t}]- \sum_{\pi\in S_t}c_\pi P_d(\pi) \right)\\
&=\Pi_{{\rm sym}}^{(t)} \left( \mathbb{E}_{\mathbb{SP}}[ U^{\otimes t}\dya{\psi_0}^{\otimes t} (U\ad)^{\otimes t}]\right)- \Pi_{{\rm sym}}^{(t)} \left( \sum_{\pi\in S_t}c_\pi P_d(\pi) \right)\\
&= \mathbb{E}_{\mathbb{SP}}[ U^{\otimes t}\dya{\psi_0}^{\otimes t} (U\ad)^{\otimes t}]- \left(\sum_{\pi\in S_t}c_\pi \right) \Pi_{{\rm sym}}^{(t)}\,.
\end{align*}
Here we have used that $ \Pi_{{\rm sym}}^{(t)}P_d(\pi)= \Pi_{{\rm sym}}^{(t)}$ for all $\pi\in S_t$ as per Eq.~\eqref{eq:inv-St}.
So, we conclude
\begin{equation}
\mathbb{E}_{\mathbb{SP}}[ U^{\otimes t}\dya{\psi_0}^{\otimes t} (U\ad)^{\otimes t}]= \left(\sum_{\pi\in S_t}c_\pi \right)\Pi_{{\rm sym}}^{(t)};
\end{equation}
taking the trace of both sides yields
\begin{equation}
\sum_{\pi\in S_t}c_\pi =\left(\Tr[\Pi_{{\rm sym}}^{(t)}]\right)^{-1}=1\,,
\end{equation}
which again recovers the result in Theorem~\ref{theo:main-SM}.

\end{proof}

\section{Proof of Theorem~\ref{theo:rank-two}}

Let us begin by recalling  Theorem~\ref{theo:rank-two}.

\begin{theorem}\label{theo:rank-two-SM}
    There exists a rank-two quantum state $\rho$ for which
    \begin{equation}
        \mathbb{E}_{\mathbb{U}}[ U^{\otimes 2}\rho^{\otimes 2}(U\ad)^{\otimes 2}]\neq \mathbb{E}_{\mathbb{SP}}[ U^{\otimes 2}\rho^{\otimes 2}(U\ad)^{\otimes 2}]\,.
    \end{equation}
\end{theorem}

\begin{proof}

Here we will show that  $\mathbb{E}_{\mathbb{U}}[ U^{\otimes 2}\rho^{\otimes 2} (U\ad)^{\otimes 2}]$ is not equal to $ \mathbb{E}_{\mathbb{SP}}[ U^{\otimes 2}\rho^{\otimes 2} (U\ad)^{\otimes 2}]$ for rank-two states. Our proof strategy is based on explicitly picking some  state, computing these averages, and showing that they are different.

Consider the following rank-two state $\rho=\lambda_0\dya{0}^{\otimes n}+\lambda_1\dya{1}^{\otimes n}$ 
such that $\lambda_0+\lambda_1=1$. Then, we have that~\cite{mele2023introduction,garcia2023deep}
\begin{align}
    \mathbb{E}_{\mathbb{U}}[ U^{\otimes 2}\rho^{\otimes 2}(U\ad)^{\otimes 2}]=&\frac{1}{d^2-1}\left(\Tr[\rho^{\otimes 2}]-\frac{\Tr[\rho^{\otimes 2} \SWAP]}{d}\right)\id\otimes \id +\frac{1}{d^2-1}\left(\Tr[\rho^{\otimes 2} \SWAP]-\frac{\Tr[\rho^{\otimes 2}]}{d}\right) \SWAP\nonumber\\
    =&\frac{1}{d^2-1}\left(1-\frac{\lambda_1^2+\lambda_2^2}{d}\right)\id\otimes \id +\frac{1}{d^2-1}\left(\lambda_1^2+\lambda_2^2-\frac{1}{d}\right)\SWAP\,.\label{eq:twirl-rank-2-U}
\end{align}
 A direct calculation also reveals that~\cite{garcia2024architectures}
 \begin{align}  
    \mathbb{E}_{\mathbb{SP}}[ U^{\otimes 2}\rho^{\otimes 2}(U\ad)^{\otimes 2}]&=\frac{(d-1)\Tr[\rho^{\otimes 2}]  -\Tr[\rho^{\otimes 2}{\rm SWAP}]  +\Tr[\rho^{\otimes 2} \Phi_s]}{d(d+1)(d-2)}\,\id_d\otimes \id_d\nonumber\\
    &+\frac{-\Tr[\rho^{\otimes 2}] + (d-1)\Tr[\rho^{\otimes 2}{\rm SWAP}]  -\Tr[\rho^{\otimes 2} \Phi_s]}{d(d+1)(d-2)} \,{\rm SWAP}\nonumber\\
    &+\frac{\Tr[\rho^{\otimes 2}]  -\Tr[\rho^{\otimes 2}{\rm SWAP}] + (d-1) \Tr[\rho^{\otimes 2} \Phi_s]}{d(d+1)(d-2)}\,\Phi_s\nonumber\\
    &=\frac{1}{d(d+1)}\,\id_d\otimes \id_d+\frac{-1 + (d-1)(\lambda_1^2+\lambda_2^2)  +2\lambda_1\lambda_2}{d(d+1)(d-2)} \,{\rm SWAP}\nonumber\\
    &+\frac{1  -(\lambda_1^2+\lambda_2^2) -2 (d-1) \lambda_1\lambda_2}{d(d+1)(d-2)}\,\Phi_s\,.\label{eq:twirl-rank-2-SP}
\end{align}
where $\Phi_s=\id_d\otimes \Omega\dya{\Phi}\id_d\otimes \Omega$ and $\ket{\Phi}=\sum_{j=1}^d\ket{jj}$. Comparing the previous result with~\eqref{eq:twirl-rank-2-U} shows that if $\lambda_1,\lambda_2\neq 0$ then $ \mathbb{E}_{\mathbb{U}}[ U^{\otimes 2}\rho^{\otimes 2}(U\ad)^{\otimes 2}]\neq \mathbb{E}_{\mathbb{SP}}[ U^{\otimes 2}\rho^{\otimes 2}(U\ad)^{\otimes 2}]$, thus showing that for rank-two states these averages can be different.

\end{proof}

\section{On the representation theory of the Brauer algebra}\label{ap:rep}

In this section we present some insights regarding the representation $F_d$ of the Brauer algebra.  

Let us start  define the propagating number ${\rm pr }(\sigma)$  as the number of ``legs'' that cross from left to right~\cite{rubey2015combinatorics} (see Fig.~\ref{fig:Brauer}). Interestingly, the propagating number allows us to group the elements of the Brauer algebra according to the value of ${\rm pr }(\sigma)$. That is, we can define as $J_t(\eta)$ all the elements of $\mathfrak{B}_t(-d)$ that have propagating number smaller or equal than $\eta$. For instance,  we have that given some $\sigma\in\mathfrak{B}_t(-d)$ which is also a permutation (i.e., $\sigma\in S_t$), then ${\rm pr}(\sigma)=t$. As such, $J_t(n-1)$ correspond precisely to all the elements in $\mathfrak{B}_t(-d)\backslash \mf{S}_t(-d)$.
Importantly, one can readily verify that the following composition table holds
\begin{equation}\label{eq:compisition}
    \noindent\begin{tabular}{c | c c }
     $\cdot$ & $S_t$ & $J_t(n-1)$   \\
    \cline{1-3}
    $S_t$ & $S_t$ & $J_t(n-1)$  \\
    $J_t(n-1)$ & $J_t(n-1)$ & $J_t(n-1)$  
\end{tabular}\,.
\end{equation}

Next, let us consider the one-dimensional representation of the Brauer algebra $r_1:\mathfrak{B}_t(-d)\rightarrow \mathbb{R}$ defined as~\cite{rubey2015combinatorics}:
\begin{equation}
    r_1(\sigma)=\begin{cases}1 \,, \quad \text{if }\sigma \in S_t\\
    0 \,, \quad \text{if }\sigma \in J_t(n-1)
    \end{cases}\,,
\end{equation}
and we can readily check that it respects the composition table of Eq.~\eqref{eq:compisition}
\begin{equation}
    \noindent\begin{tabular}{c | c c }
     $\cdot$ & $r_1(S_t)$ & $r_1(J_t(n-1))$   \\
    \cline{1-3}
    $r_1(S_t)$ & 1 & 0  \\
    $r_1(J_t(n-1))$ & 0 & 0  
\end{tabular}\,.
\end{equation}

From here, we decompose the $t$-th fold tensor product of the Hilbert space as
\begin{equation}
    \HC^{\otimes t}=\HC^{\otimes t}_{{\rm sym}}\oplus \HC_{\perp}\,,
\end{equation}
where $\HC^{\otimes t}_{{\rm sym}}$ denote the symmetric subspace of $\HC^{\otimes t}$. I.e., for any $\ket{\Psi}\in \HC^{\otimes t}_{{\rm sym}}$, one has that $\Pi_{{\rm sym}}^{(t)}\ket{\Psi}=\ket{\Psi}$, and hence from Lemma~\ref{lem:no-overlap} we obtain
\begin{equation}
    F_d(\sigma)\ket{\Psi}=\begin{cases}1\ket{\Psi} \,, \quad \text{if }\sigma \in S_t\\
    0\ket{\Psi}=0 \,, \quad \text{if }\sigma \in J_t(n-1)
    \end{cases}\,.
\end{equation}
As such, this directly implies that the action of the representation $F_d$ induces a decomposition into irreducible representation of the symmetric subspace of the Hilbert space as
\begin{equation}
    \HC^{\otimes t}_{{\rm sym}}=\id_{d_{{\rm sym}}}\otimes r_1\,,
\end{equation}
where $\id_{d_{{\rm sym}}}$  is an identity of size $d_{{\rm sym}}=\dim(\HC^{\otimes t}_{{\rm sym}})=d(d+1)\cdots (d+t-1)$.

The previous irrep decomposition of the symmetric subspace of $\HC^{\otimes t}$ is fundamental for our results as the pure states that we  average over belong precisely to $\HC^{\otimes t}_{{\rm sym}}$ where the elements of $J_t(n-1)$ act as zero, and this eventually leads to them not contributing in the averages.

\section{Statistical indistinguishability between $\mathrm{U}(d)$ and $\mathbb{SP}(d/2)$ state distributions}\label{ap:indist}
 
The fact that, for any $t \in \mathbb{N}$, the uniform distribution over symplectic states forms a unitary $t$-design implies that no quantum experiment can distinguish between states sampled uniformly from $\mathrm{U}(d)$ and those sampled from $\mathbb{SP}(d/2)$, even if the sampled state is queried an arbitrarily large number of times.
In fact, by the definition of a $t$-design, the uniform probability distribution over $t$ copies of symplectic states is identical to the uniform distribution over $t$ copies of unitary states. Specifically, the associated density matrices
\begin{align}
\rho_{\mathrm{U}} \coloneqq \mathbb{E}_{\psi \sim \mathrm{U}(d)} \left[ (|\psi\rangle \langle \psi|)^{\otimes t} \right]\,, 
\quad \text{and} \quad 
\rho_{\mathbb{SP}} \coloneqq \mathbb{E}_{\psi \sim \mathbb{SP}(d/2)} \left[ (|\psi\rangle \langle \psi|)^{\otimes t} \right]\,,
\end{align}
are identical. As a result, the two distributions over states are indistinguishable, even when arbitrary POVMs are performed. Thus, we can conclude the following:

\indist*

\section{Random symplectic circuit in a one-dimensional lattice}\label{ap:brick}

Let us here consider a one-dimensional layered circuit $U$ where local two-qubit gates act on neighboring pairs of qubits in a
brick-like fashion. As such,  $U$ can be expressed as 
\begin{equation}\label{eq:circuit}
    U=\prod_{l=1}^L U_{l}\,,
\end{equation}
and where
\begin{equation}\label{eq:layer}
 U_{l}= \prod_{j = 1}^{\frac{n}{2}} U_{2j-1,2j}^l\prod_{j = 1}^{\frac{n-2}{2}} U_{2j,2j+1}^l\,.
\end{equation}

Then, let us assume that each $U_{i,j}$ is sampled according to the Haar measure over some local group $G_{i,j}\subseteq\mathbb{U}(4)$, and let us define as $\EC_L$ the distribution of unitaries obtained from an $L$-layered circuit $U$. In particular, it is known that in the large-$L$ limit, $\EC_L$ will converge to some global group $G\subseteq\mathbb{U}(2^n)$ which depends on the choice of the  $G_{i,j}$'s. Here, one can quantify how many layers are needed for $\EC_L$  to form  an approximate $t$-design over $G$ by studying the spectral gap $\lambda$ of the second-order moment operator of a single layer
\begin{align}
    \TC_{\EC_1}^{(2)}&= \underset{U \sim \EC_1}{\mathbb{E}}[U^{\otimes 2} \otimes (U^*)^{\otimes 2}] = \int_{\EC_1}dU U^{\otimes 2}\otimes (U^*)^{\otimes 2} \label{eq:tiwrlG}\,.
\end{align}
The spectral gap $\lambda$ is defined as the largest $<1$ eigenvalue of $\TC_{\EC_q}^{(2)}$~\cite{brandao2016local}. Specifically,  $\EC_1$ forms an $\varepsilon$-approximate $2$-design if $\lambda^L\leq\varepsilon/2^n$. That is, when
\begin{equation}
    L\sim\frac{\log(\frac{1}{\epsilon})+n\log(2)}{\log(\frac{1}{\lambda})}\,.
\end{equation}
Crucially, this means that if each local gate has $N$ parameters, the $L$-layered circuit will have a total number of parameters $N_T$ scaling as
\begin{equation}
    N_T\sim  \frac{N(n-1)\left(\log(\frac{1}{\epsilon})+n\log(2)\right)}{\log(\frac{1}{\lambda})}\,,
\end{equation}
where we have also used the fact each layer has $(n-1)$ gates as per Eq.~\eqref{eq:layer}.

In what follows we will consider two different cases. In the first, $G_{i,j}=\mathbb{SU}(4)$ for all $i,j$, in which case the circuit's distribution converges to $\mathbb{U}(2^n)$. Here, each local gate is parametrized by $15$ parameters.  In the second, we take  $G_{j,j'}^l=\mathbb{SO}(4)$ if $j,j'$ are not equal to one, and $G_{j,j'}^l=\mathbb{SP}(2)$ if  $j$ or $j'$ are equal to one.  In this case one can show that  $G=\mathbb{SP}(2^n/2)$~\cite{garcia2024architectures}. Now, the local orthogonal gates have $6$ parameters while the symplectic ones have $10$. Since most of the gates are orthogonal, we will assume that each gate has, on average, $6$ parameters, as this statement becomes true in the large $n$ limit. We refer to the previous two cases as the unitary and symplectic case, respectively. 

As shown in~\cite{deneris2024exact}, for the unitary case we find $\lambda=0.64$, while for the symplectic case $\lambda=0.6461$. A direct calculation reveals that the ratio of total number of parameters needed for the unitary and symplectic circuit to be an $\epsilon$-approximate $2$-design over their respective group --respectively denoted as $N_{\mathbb{U}}$ and $N_{\mathbb{SP}}$-- is
\begin{equation}
   \frac{N_{\mathbb{SP}}}{N_{\mathbb{U}}}= \frac{\log(\frac{1}{.64})}{\log(\frac{1}{.6461})}\frac{6}{15}\approx .40\,.
\end{equation}
  
From the previous, we find that the following result holds.
\setcounter{result}{1}
\begin{result}\label{res:params-SM}
A random circuit in a one-dimensional topology composed of symplectic random two-qubit gates can form an $\epsilon$-approximate  $2$-design with $~60\%$ less parameters than a circuit with the same topology and with unitary random two-qubit gates. 
\end{result}
\section{Classical shadows with symplectic unitaries} \label{ap:shad}
In the appendix we (briefly) review the formalism of \textit{classical shadows}~\cite{huang2020predicting}, establishing the details necessary to prove Result~\ref{result:shad}, that symplectic shadows simply reproduce the known unitary shadows scheme. Classical shadow tomography refers to a randomized measurement scheme for determining properties of an unknown quantum state $\rho$ from very few measurements. A specific classical shadows protocol is defined by a choice of unitary ensemble $\mcu$ and measurement basis $\mcw=\{\ket{w}\}_w$. One proceeds by applying unitaries sampled from $\mcu$ to copies of  $\rho$, and then measuring in the basis $\mcw$. Upon sampling a given $U\in\mcu$ and subsequently measuring a particular $\ket{w}\in\mcw$ one stores the state $U^\dagger\ketbra{w}U$, the net effect of which is to implement the quantum channel
\begin{align}
\mcm(\rho)&=   \sum_w\int_{U\sim\hspace{0.5mm}\mcu}\Tr[U\rho U^\dagger\ketbra{w}] U^{\dagger }\ketbra{w}U^{}\nonumber\\
&=  \Tr_1\left[\left(\rho \otimes \id\right) \sum_w\int_{U\sim\hspace{0.5mm}\mcu} U^{\dagger\otimes 2}\ketbra{w}^{\otimes 2}U^{\otimes 2}\right]\label{eq:mc}\,.
\end{align}
For a given obtained pair $(U,w)$ one additionally defines a \textit{classical shadow} $\hat{\rho}=\mcm^{-1}(U^\dagger\ketbra{w}U)$ of $\rho$. In general the measurement channel $\mcm$ will fail to be invertible, and the pseudo-inverse is instead used. The utility of the classical shadows lies in their capacity to be used as a proxy for $\rho$ itself for the purposes of estimating expectation values of operators that lie within the image of $\mcm$; indeed we have
\begin{align*}
\mbe_{U,w\hspace{0.5mm}|\hspace{0.5mm}\rho} \Tr[\hat{\rho}O]&= \mbe_{U,w\hspace{0.5mm}|\hspace{0.5mm}\rho} \Tr[\mcm^{-1}(U^\dagger\ketbra{w}U)O]\\
&=\mbe_{U,w\hspace{0.5mm}|\hspace{0.5mm}\rho} \Tr[U^\dagger\ketbra{w}U\mcm^{-1}(O)]\\
&= \sum_w\int_{U\sim\hspace{0.5mm}\mcu}\Tr[U\rho U^\dagger\ketbra{w}]  \Tr[U^\dagger\ketbra{w}U\mcm^{-1}(O)]\\
&= \Tr[\rho\sum_w\int_{U\sim\hspace{0.5mm}\mcu} U^\dagger\ketbra{w}U  \Tr[U^\dagger\ketbra{w}U\mcm^{-1}(O)]]\\
&=\Tr[\rho(\mcm\circ\mcm^{-1})(O)]\\
&=\Tr[\rho\hspace{0.5mm}{\rm proj}_{\mathfrak{Im}(\mcm)}(O)]\,.
\end{align*}
So for operators within the ``visible space'' $\mathfrak{Im}(\mcm)$ of the protocol, the shadows reproduce in expectation the correct value. The remaining pertinent question is that of how many  shadows one expects to need before being confident that the obtained empirical average  closely matches the true value, which is controlled by the variance of the estimators. This too may be readily calculated, via the standard expression
\begin{equation*}
{\rm Var}[\hat{o}_i ]=\mbe[\hat{o}_i^2]-\mbe[\hat{o}_i]^2.
\end{equation*}
These two expectations may in turn be written as
\begin{align}
\mbe[\hat{o}^2] &= \sum_w\int_{U\sim\hspace{0.5mm}\mcu}   \Tr\left[\rho U^\dagger\Pi_wU\right]\Tr\left[O \mcm^{-1}\left(U^\dagger\Pi_wU\right)\right] ^2 \nonumber \\
&=\sum_w  \Tr\bigg[\left(\rho\otimes \mcm^{-1}(O)\otimes  \mcm^{-1}(O)\right) \int_{U\sim\hspace{0.5mm}\mcu}  U^{\dagger\otimes 3}\Pi_w^{\otimes 3}U^{\otimes 3} \bigg] \,,\label{eq:ea2}
\end{align}
and
\begin{align}
\mbe[\hat{o}]^2 &=\left( \sum_w\int_{U\sim\hspace{0.5mm}\mcu}   \Tr\left[\rho U^\dagger\Pi_wU\right]\Tr\left[O\mcm^{-1}\left(U^\dagger\Pi_wU\right)\right]\right) ^2 \nonumber \\
&= \left(\sum_w  \Tr\left[\left(\rho\otimes  \mcm^{-1}(O)\right)\int_{U\sim\hspace{0.5mm}\mcu}  U^{\dagger\otimes 2}\Pi_w^{\otimes 2}U^{\otimes 2} \right] \right) ^2 \label{eq:e2a}\,.
\end{align}
With the help of these expressions we are now in a position to see that the proof of Result~\ref{result:shad} is fairly immediate. 

\shad*
\begin{proof}

From the characterization Eqs.~\eqref{eq:mc},~\eqref{eq:ea2} and~\eqref{eq:e2a} of classical shadow protocols as relying solely on the second and third state-moments, 
the result now readily follows from the fact that the unitary symplectic matrices form unitary state 2- and 3-designs (and immediately applies to any ensemble with this property). For example,  the measurement channels are identical:
\begin{align*}
  \mcm_{\mathbb{SP}(d/2)}(A)&=  \Tr_1 \left[(A \otimes \id)\sum_w\int_{U\in \mathbb{SP}(d/2)}   U^{\dagger\otimes 2}\Pi_w^{\otimes 2}U^{\otimes 2}\right] \\
  &=  \Tr_1 \left[(A \otimes \id)\sum_w\int_{U\in \mathbb{U}(d)}   U^{\dagger\otimes 2}\Pi_w^{\otimes 2}U^{\otimes 2}\right] \\
  &=  \mcm_{\mathbb{U}(d)}(A)\,,
\end{align*}
where we have used the fact that $\mathbb{SP}$ is a unitary state 2-design. The equality of the terms in the expressions Eqs.~\eqref{eq:ea2} and~\eqref{eq:e2a} for the variance of the estimators follows with a similar immediacy.

\end{proof}
\end{document}